\documentclass[runningheads]{llncs}
\usepackage[T1]{fontenc}
\usepackage{graphicx}
\usepackage{subcaption} 
\usepackage{tikz}
\usetikzlibrary{spy}
\usepackage{amsmath}
\usepackage{amssymb}
\usepackage{bm}

\usepackage{algorithmicx}
\usepackage[ruled,vlined,linesnumbered]{algorithm2e}

\newcommand{\NN}{\mathbb{N}}

\newcommand{\RR}{\mathbb{R}}

\newcommand{\R}{\mathbb{R}}

\newcommand{\vertL}{\:\: | \:\:}

\newcommand{\dia}{\textnormal{dia}}

\newcommand{\ST}{\mathcal{S}}
\newcommand{\RC}{\mathcal{R}}
\newcommand{\RO}{\mathcal{O}}
\newcommand{\SK}{\Sigma}

\begin{document}
\title{Skeletonisation Scale-Spaces}
\titlerunning{Skeletonisation Scale-Spaces}
%

\author{Julia Gierke and Pascal Peter}

\authorrunning{J. Gierke and P. Peter}
%
\institute{Mathematical Image Analysis Group,
	Faculty of Mathematics and Computer Science,\\ Campus E1.7,
	Saarland University, 66041 Saarbr\"ucken, Germany.\\
	\{gierke,peter\}@mia.uni-saarland.de}
\maketitle              
\begin{abstract}

The medial axis transform is a well-known tool for shape recognition. Instead of the object contour, it equivalently describes a binary object in terms of a skeleton containing all centres of maximal inscribed discs. While this shape descriptor is useful for many applications, it is also sensitive to noise: Small boundary perturbations can result in large unwanted expansions of the skeleton. Pruning offers a remedy by removing unwanted skeleton parts.
In our contribution, we generalise this principle to skeleton sparsification: We show that subsequently removing parts of the skeleton simplifies the associated shape in a hierarchical manner that obeys scale-space properties. To this end, we provide both a continuous and discrete theory that incorporates architectural and simplification statements as well as invariances. We illustrate how our skeletonisation scale-spaces can be employed for practical applications with two proof-of-concept implementations for pruning and compression.

\end{abstract}
\section{Introduction}

Classical scale-space theory is closely connected to shape analysis \cite{SP08}: In particular, morphological scale-spaces   \cite{AGLM93,BM92,CS96,KS96,ST93,BS94} consider the evolution of object boundaries under operations derived from dilation and erosion \cite{So99a}. However, the silhouette is not the only way to describe a shape. Blum~\cite{Bl67} has proposed the medial axis transform (MAT) as an equivalent shape descriptor motivated from the grass fire analogy: Propagating fronts originating from the object contour meet in the middle of the object, creating shocks. These locations define the skeleton of the shape, from which the full object can be reconstructed. Multi-scale ideas have also been explored for the MAT and related concepts with a primary focus on stability under noisy shape boundaries \cite{PSSDZ03,SBS16}.

Recently, C\'ardenas et al.~\cite{CPW19} have used image sparsification to define scale-spaces: They gradually remove image pixels and reconstruct the evolving image from this sparse data. This resembles skeleton pruning~\cite{SB98}, which removes unwanted skeleton parts arising for instance from boundary noise. 
Leveraging this similarity, we establish novel connections between scale-space and MAT theory. Our goal is to generalise the concept of skeleton pruning in a formal scale-space setting by translating the sparsification paradigm from images to the medial axis.

\textbf{Our Contributions.} To this end, we propose a scale-space framework for shape evolutions that arise from sequential sparsification of skeletons. Inspired by the work of C\'ardenas et al.~\cite{CPW19}, we describe the order in which skeleton points are removed in terms of a sparsification path. Based on a  small amount of requirements, we establish architectural, simplification, and invariance properties for our scale-spaces. With a compression and a pruning example, we demonstrate how the sparsification path can be adapted to different practical applications. In particular, by imposing task-specific requirements on the sparsification path, we achieve additional theoretical guarantees for each application. 

\textbf{Related Work.} Our scale-space theory relies on the axiomatic frameworks that were developed for classical scale-spaces based on partial differential equations (PDEs) \cite{AGLM93,Ii62,Li11,SchW98,We97}, in particular on the work of Alvarez et al.~\cite{AGLM93}. For a comprehensive overview we refer to Weickert~\cite{We97}.
Due to their focus on shape analysis, morphological scale-spaces~\cite{AGLM93,BM92,CS96,KS96,ST93,BS94} are conceptually closer to our approach. However, they analyse the evolution of the shape boundary, while we consider the medial axis to characterise shapes.

Deriving a scale-space from a shrinking skeleton is closely related to the sparsification scale-spaces of C\'ardenas et al.~\cite{CPW19}, who remove image pixels successively instead of skeleton points and reconstruct with inpainting. For the medial axis, this can be seen as a form of pruning~\cite{BLL07,ML12,Og94,SB98,SS16,TH02}. Ogniewicz~\cite{Og94} proposed a scale-space-inspired hierarchic pruning similar to our branch pruning in Section~\ref{sec:applications}. It is however practice-oriented, less general, and includes no systematic investigation of scale-space properties. 

Besides the concept of hierarchy in pruning, connections between MAT and scale-space theory are rare. The multi-scale medial axis of Pizer et al.~\cite{PBCFM94} consist of curves in scale-space representing middle points. Similarly, image cores \cite{PEFM98} generalise the MAT based on scale-space theory to gain invariances. Tsogkas and Kokkinos~\cite{TK12} have used multi-scale concepts in a MAT-inspired learning approach to compute symmetry axes. 

\textbf{Organisation.}
Section~\ref{sec:skeletonisation} provides the theoretical background required for our new scale-spaces in Section~\ref{sec:skelscale}, which we adapt to specific applications in Section~\ref{sec:applications}. We conclude with a discussion and outlook on future work in Section~\ref{sec:conclusion}.

\section{Review: Skeletonisation}
\label{sec:skeletonisation}

For our skeletonisation scale-spaces, we require definitions of objects, distances, and the medial axis transform (MAT) \cite{Bl67}. We consider a binary image as a  function $f: \Omega \rightarrow \{0,1\}$
that maps the image domain $\Omega \subset \R^2$ to the background value $0$ and the object value $1$. The object is defined by set $O=\{ \bm x \in \Omega \, | \, u(\bm x) = 1\}$ and $\partial O$ denotes its boundary. We assume that $O$ is bounded and closed, such that its boundary belongs to the set. Hence, $O$ is a compact set. Furthermore, let $d(\bm x, \bm y)$ denote the Euclidean distance for $\bm x, \bm y \in \Omega$.

\begin{definition}[Distance Map $D$]\label{def:dm}
	The \textit{distance map} denotes the minimal distance to the boundary $\partial O$ according to
	\begin{equation}
		D: O \rightarrow \RR, \quad \bm x \mapsto \min\limits_{\bm y \in \partial O} d(\bm x,\bm y)
	\end{equation}
\end{definition}

The distance map allows us to define $B_{D(\bm x)}(\bm x):=\{\bm y \in \Omega \,  | \, d(\bm x, \bm y) \leq D(\bm x) \}$ as the largest disc around $\bm x$ that is fully contained in $O$. With such \textit{inscribed discs} we can define the medial axis.

\begin{definition}[Skeleton $\SK$]\label{def:skel}
	The \textit{skeleton} or \textit{medial axis (MA)} $\SK$ is defined as the set of centres of inscribed discs of maximal radius:
	\begin{equation*}   
		\Sigma := \{ \bm x \in O \, \vertL \, \forall \bm y \in O : B_{D(\bm x)}(\bm x) \not\subset B_{D(\bm y)}(\bm y) \} \, .
	\end{equation*}
\end{definition}

In the definition, we call the disc around $\bm x$ \textit{maximal}, since it is not fully contained in any other larger disc $B_{D(\bm y)}(\bm y)$ around  a skeleton point $\bm y$. Equivalently, $\SK$ is the set of skeleton points with at least two equidistant boundary points. 

\begin{definition}[Skeleton, Reconstruction, and Object Transform]\label{def:transforms}
	The \textit{skeleton transform} $\ST(f)=\Sigma$ maps the image $f$ to its skeleton $\SK$. The \textit{reconstruction transform} $\RC(\Sigma)=f$ is defined via the object transform 
	\begin{equation}
	\RO (\SK) := \bigcup_{\bm x \in \SK}  B_{D(\bm x)}(\bm x) \, .
	\end{equation}
\end{definition}

Since $\RO(\SK)$ is a continuous mapping of the compact set $\SK$, all reconstructions are compact as well and $\RO(\SK) = O$. In addition to this continuous case, we also consider discrete images with resolution $n_x \times n_y$ written as vectors $\bm f \in \{0,1\}^{n}$ with the discrete image domain $\Omega_D=\{1,...,n\}$. With the Euclidean distance $d(i,j)$ between the pixel centres of $i,j \in \Omega_D$ and discrete discs, we transfer all other definitions from the continuous to the discrete setting. However, using discrete discs in Definition~\ref{def:transforms} can lead to $\RO(\SK)  \neq O$ and thus $\RC(\Sigma)\neq f$. We discuss the consequences of these observations for our scale-spaces in Section~\ref{sec:properties}.

\section{Skeletonisation Scale-Spaces}
\label{sec:skelscale}

Inspired by the work of  C\'ardenas et al.~\cite{CPW19}, we define skeletonisation scale-spaces as the family of images emerging from successively removing parts of an initial skeleton. The order of this pruning is defined by \textit{sparsification paths}.

\begin{definition}[Sparsification Path]\label{def:sparspath}	
Given a skeleton $\SK$, a \textit{sparsification path} $P = (P_1,\dots,P_{m})$ with $m \in \NN \setminus \{0\}$ is an ordered collection of non-empty sets that form a partition of $\SK$, i.e. for all $s \neq t$: $P_s \neq \emptyset$, $P_s \cap P_t = \emptyset$, and $\bigcup_{\ell=1}^{m}P_\ell = \SK$. 
\end{definition}

In Section~\ref{sec:applications}, we show that the sparsification path is a design tool for task-specific scale-spaces. Given the path, we can define the scale-space itself.

\begin{definition}[Skeletonisation Scale-Space] \label{def:skelspace}
	Consider the binary image $f: \Omega \rightarrow \{0,1\}$ with domain $\Omega$, skeleton $\SK = \ST(f)$, reconstruction $\RC(\SK)$ and sparsification path $P=(P_1,...,P_m)$ partitioning $\SK$, $m \in \NN \setminus \{0\}$.
	The \textit{skeletonisation scale-space} is the family $(u_\ell, \SK_\ell)_{\ell=0}^{m}$ of images $u_\ell$ and skeletons $\SK_\ell$ obeying
	\begin{enumerate}
		\item $\SK_0 := \SK = \ST(f)$,
		\item $\SK_\ell := \SK_0 \setminus \bigcup_{i=1}^{\ell} P_i$
	    \hspace{23mm} \textnormal{for} $\ell \in \{1,\dots,m\}$,
		\item $u_\ell := \RC(\SK_\ell)$ \textnormal{and} $O_\ell =\RO(\SK_\ell)$ \hspace{7mm}  \textnormal{for} $\ell \in \{0,\dots,m\}$.
	\end{enumerate}
\end{definition}

At every scale $\ell$, we obtain a skeleton-image-pair by successively removing $P_\ell$ from the skeleton according to the sparsification path and reconstruction $u_\ell$ from the pruned skeleton. We grant full flexibility regarding the choice of the sets $P_\ell$ and number $m$ of sparsification steps. It implies the $m+1$ discrete scales.

So far, our definition applies to space-continuous skeletons. Analogously, we define fully discrete skeletonisation scale-spaces: We replace the continuous image $f$ by its discrete counterpart $\bm f \in \R^n$ with $n = n_x \cdot n_y$ pixels and use the discrete definitions of skeletons and the associated transforms from Section~\ref{sec:skeletonisation}. Then, all other definitions for sparsification path and scale-space carry over directly. In Section~\ref{sec:properties}, we consider scale-space properties primarily for the continuous formulation. In case those properties do not apply to the discrete setting, we discuss these differences explicitly. 

\subsection{Generalised Scale-Space Properties}
\label{sec:properties}

We investigate skeletonisation scale-spaces in the sense of Definition~\ref{def:skelspace} from the perspective of scale-space research. Specifically, we show how properties established by Alvarez et al.~\cite{AGLM93} transfer to this new class of scale-space.

\medskip
\noindent
{\bf Property 1: (Original Skeleton as Initial State).}\\[1mm]
By Definition~\ref{def:skelspace}, at scale $\ell=0$, the skeleton $\SK_0 = \ST(f) = \SK$ is the skeleton $\SK$ of the original object $O$ contained in the image $f$. In the continuous setting, the image $u_0$ is thus the original image $f$ due to $u_0 = \RC(\ST(f)) =f$. For the discrete case, the reconstruction transform $\RC$ is only an approximation, not an inverse of $\ST$ and thus only $\SK_0 = \SK$ holds.

\medskip
\noindent
{\bf Property 2: (Causality).}\\[1mm]
It does not matter if the sparsification path $P = (P_1,\dots,P_m)$ is traversed from the initial scale to scale $\ell$ or from any intermediate scale: The resulting skeleton-image-pair at scale $\ell$ remains the same. According to Definition~\ref{def:skelspace}, for $\ell \in \{0,...,m\}$ and $k \in \{0,...,\ell-1\}$, we have
		\begin{equation}
			\SK_\ell = \SK_0 \setminus \bigcup_{i=1}^{\ell}P_i = \left(\SK_0 \setminus \bigcup_{i=1}^{k}P_i \right) \setminus \bigcup_{i=k+1}^{\ell}P_i =  \SK_{k} \setminus \bigcup_{i=k+1}^{\ell}P_i \, .
		\end{equation}
Thus, we can reach scale $\ell$ equivalently in $\ell$ steps from scale $0$ or in $\ell-k$ steps from scale $k$. Since the skeleton at scale $\ell$ remains the same, this also holds for the image $u_\ell = \ST(\SK_\ell)$.

\medskip
\noindent
{\bf Property 3: (Lyapunov Sequences).}\\[1mm] 
Lyapunov sequences quantify the simplification of the shape. In the following, we specify generic sequences that apply to all skeletonisation scale-spaces. In Section~\ref{sec:applications}, we show that for specific tasks, we can find meaningful specialised Lyapunov sequences by imposing additional requirements.

\begin{proposition}[Decreasing Object Area] 
	\label{prop:area}
	\noindent The area $a_\ell := |O_\ell|$ of the object decreases as the scale parameter $\ell$ increases, i.e. $a_\ell \leq a_k$ for $k > \ell$.
\end{proposition}

\begin{proof}
	For the path $P = (P_1,\dots,P_m)$ we have $\SK_\ell = \SK_{\ell+1} \cup P_\ell$ and thus
	\begin{equation*}
		\quad \bigcup_{\bm x \in \SK_\ell} B_{D(\bm x)}(\bm x)\quad = \quad \bigcup_{\bm x \in \SK_{\ell+1}} B_{D(\bm x)}(\bm x) \:\cup\: \bigcup_{\bm x \in P_\ell} B_{D(\bm x)}(\bm x) \, .
	\end{equation*}
	Therefore, we can derive
	\begin{equation*}
		\RO(\SK_\ell) = \quad \bigcup_{\bm x \in \SK_\ell} B_{D(\bm x)}(\bm x) \:\supseteq \bigcup_{\bm x \in \SK_{\ell+1}} B_{D(\bm x)}(\bm x) \quad = \RO(\SK_{\ell+1}) \, .
	\end{equation*}
	Thus, we can conclude $a_\ell = | \RO(\SK_\ell) |  \leq  | \RO(\SK_{\ell+1}) | = a_{\ell +1}$.  \qed
\end{proof}

Similarly, we can also consider the spatial dimensions of the object. In the following, the diameter $\dia(S)$ of a closed set $S \subset \Omega$ is defined as
\begin{equation*}
	\dia(S) := \max\{ d(\bm x, \bm y) \vertL \bm x, \bm y \in S \},
\end{equation*}
where $d(i,j)$ is the continuous or discrete distance function. Then, the diameter $\dia(\RO(\SK_\ell))$ of the reconstructed object defines a Lyapunov sequence.

\begin{proposition}[Decreasing Object Diameter] 
	The object diameter decreases with increasing scale $\ell$, i.e. for $\ell \geq 0$, we have
	\begin{equation}
		\dia(\RO(\SK_\ell)) \geq \dia(\RO(\SK_{\ell+1})).
	\end{equation}
\end{proposition}

\begin{proof}
	 Assume that $\dia(\RO(\SK_\ell)) < \dia(\RO(\SK_{\ell+1}))$.  Then, there are $\bm x, \bm y \in \RO(\SK_{\ell+1})$ with $d(\bm x, \bm y) > \dia(\RO(\SK_\ell))$. However, from the proof for  Proposition~\ref{prop:area}, we know  $\RO(\SK_{\ell+1}) \subseteq \RO(\SK_{\ell})$ and hence also $\bm x, \bm y \in \RO(\SK_{\ell})$. This is a contradiction to the definition of the diameter and our assumption. Thus, the proposition holds. \qed
\end{proof}

\noindent
{\bf Property 4: (Invariance).}\\[1mm] 
At each scale $\ell$, the skeleton $\SK_\ell$ uniquely defines the skeleton-image-pair. Therefore, the scale-space inherits all invariances that apply to skeletons. This includes translations, rotations, scaling, and mirroring. For each such transformation $\mathcal{T}$ we have $\ST(\mathcal{T}(f)) = \mathcal{T}(\ST(f))$. 
In the discrete setting this holds only for on-grid transformations that preserve the shape. This includes on-grid translations, $90^\circ$ rotations, and mirroring along image axis.

\medskip
\noindent
{\bf Property 5: (Empty Image as Steady State).}\\[1mm]
The image for the coarsest scale is defined by $\SK_{m} = \SK_0\setminus \bigcup_{i=1}^{m} P_i$. According to Definition~\ref{def:sparspath}, $P_1,\dots,P_m$ form a partition of $\SK$.  Furthermore, Definition~\ref{def:skelspace} implies $\SK_0 = \SK$. Thus, we have $\SK_m =  \SK \setminus \SK = \emptyset$. The corresponding reconstruction is thus an empty image.

\section{Task-Specific Skeletonisation Scale-Spaces}
\label{sec:applications}

So far, our generic scale-space framework only demands very few requirements. With the sparsification path from Definition~\ref{def:sparspath}, we can trade stricter requirements for application-specific scale-space properties. Just as sparsification scale-spaces~\cite{CPW19}, our framework allows \textit{uncommitted} strategies that prune the skeleton without dependency on the input shape. However, \textit{committed} approaches are often more interesting: They adapt to the evolving shape according to  desirable properties for practical use. To illustrate this principle, we provide two examples of task-specific skeletonisation scale-spaces. 
We also show results for an implementation that uses the maximal disc thinning algorithm MDT \cite{PB12} for skeleton computation. Therefore, we directly consider the discrete setting in the following. 

\subsection{Skeletonisation Scale-Spaces for Compression}
\label{sec:comp}

First, inspired by the work of Mühlhaus~\cite{Mü20}, we propose a MAT-based scale-space for shape compression. It reduces the amount of skeleton points to a very sparse subset while keeping the shape reconstruction as exact as possible. Since we are only interested in the scale-space aspect of such an application, we do not discuss actual coding here and refer to Mühlhaus~\cite{Mü20} instead. Our goal is to remove points in ascending order of their influence on the reconstruction. To this end, we define the reconstruction impact of a given set $S \subset \SK$ as the part of the object that can only be reconstructed by its skeleton points.

\begin{definition} [Reconstruction Impact]
	For a set of skeleton points $S \subseteq \SK_\ell$ and $\ell \in \{0,\dots,m\}$ we define the \textit{reconstruction impact} $I_{\ell,S}$ by
	\begin{equation}
		I_{\ell,S} := O_\ell \setminus \Bigl( \bigcup_{i \in \SK_\ell \setminus S} B_{D(i)}(i) \Bigr).
	\end{equation}
\end{definition}
\noindent
Removing $S$ from the skeleton implies removing $I_{\ell,S}$ from the object, i.e.
\begin{equation}
	\label{eq:removeimpact}
	O_{\ell+1} = O_\ell \setminus I_{\ell,S}
\end{equation}
Based on the reconstruction influence, Algorithm~\ref{alg:skelcomp} defines a skeleton compression path that minimises the reconstruction error. For the sake of a convenient notation, we define $I_{\ell,i} := I_{\ell,\{i\}}$ for a given pixel index $i \in \SK$. 
\begin{algorithm}[t]
	\caption{Skeleton Compression Path\label{alg:skelcomp}}
	$\ell \gets 0$, $r \in \NN$ user parameter (points to remove per step) \; 
	\While{$|\SK_\ell| > 0$}{
		$k \gets |\SK_\ell|$, $s \gets \min(r,|\SK_\ell|)$\;
		\tcc{Order all skeleton points by reconstruction influence.}
		Let $\{c_1,...,c_k\} = \SK_\ell$ with $I_{\ell,c_i} \leq I_{\ell,c_j}$ for $i \leq j$\;
		\tcc{Remove $s$ points with smallest impact on reconstruction.}
		$P_{\ell+1} \gets \{c_1,...,c_s\}$, $\SK_{\ell+1} \gets \SK_\ell \setminus P_{\ell+1}$, $\ell \gets \ell + 1$\;
	}	
\end{algorithm}
To demonstrate that this path indeed fulfils our compression goals, we consider the \textit{relative reconstruction error} $\mathcal{E}_\ell$. This ratio of missing object points and total object points is a
simplified scale-space version of the exactness of reconstruction criterion for skeleton quality~\cite{PB12}:
\begin{equation}
	\mathcal{E}_\ell := \frac{|O_0| - |O_\ell|}{|O_0|} \, .
\end{equation}
Note that this captures all possible reconstruction errors since removing skeleton points can only reduce the reconstructed area according to Proposition~\ref{prop:area}. Since $|O_\ell|$ is decreasing for coarser scales, the relative error increases and thus constitutes another Lyapunov sequence as a corollary to Proposition~\ref{prop:area}. Beyond that, we can also guarantee a minimal increase in error as desired for compression.

\begin{proposition}[Minimal Relative Error Increase]
Among all possible sparsification paths, for all $\ell \in \{0,...,m-1\}$, the increase in relative error $\mathcal{E}_{\ell+1}-e_\ell$ is minimal for a skeleton compression scale-space.
\end{proposition}

This holds directly since $|O_{\ell+1}|=|O_{\ell}|-|I_{\ell,P_{\ell+1}}|$ according to Eq.~\eqref{eq:removeimpact}. 
By definition, Algorithm~\ref{alg:skelcomp} selects $P_{\ell+1}$ such that $I_{P_{\ell+1}}$ is minimised over all possible choices of $P_{\ell+1}$. Thus, a skeleton compression scale-space guarantees that at each scale we have a sparse skeleton with the smallest possible relative error.

Another quality measure for the discrete medial axis is skeleton minimality \cite{PB12}. It relates the number of object points to the number of skeleton points needed to represent the object by
\begin{equation}
	\mathcal{M}_\ell :=  \frac{|\SK_\ell|}{|O_\ell|} \, .
\end{equation}
For compression purposes, skeleton minimality should be as small as possible, which implies that only a small amount of skeleton points needs to be stored to reconstruct a large object area. For our skeleton compression scale-space, this is reflected by a corresponding Lyapunov sequence.

\begin{proposition}[Skeleton Minimality is a Lyapunov Sequence]
	Skeleton minimality decreases with increasing scale $\ell$, i.e. $\mathcal{M}_\ell \geq \mathcal{M}_{\ell+1}$.
\end{proposition}

\begin{proof}
First, we decompose $\mathcal{M}_\ell$ according to the sparsification path from Definition~\ref{def:skelspace} and the reconstruction impact from Eq.~\eqref{eq:removeimpact}. This yields
\begin{equation}
	\label{eq:ml}
	\mathcal{M}_\ell = \frac{|\SK_\ell|}{|O_\ell|} = \frac{|\SK_{\ell+1}|+|P_{\ell+1}|}{|O_{\ell+1}|+|I_{\ell, P_{\ell+1}|}}
\end{equation}	
Furthermore, due to the sorting in Algorithm~\ref{alg:skelcomp}, the average unique area of the removed points in $P_{\ell+1}$ is smaller or equal to the average unique area of the remaining skeleton points in $\SK_{\ell+1}$, i.e.
\begin{equation}
	\label{eq:avgarea}
\frac{|I_{\ell, P_{\ell+1}}|}{|P_{\ell+1}|} \leq \frac{|I_{\ell, \SK_{\ell+1}}|}{|\SK_{\ell+1}|} \leq \frac{|O_{\ell+1}|}{|\SK_{\ell+1}|} \, .
\end{equation}
Combining both Eq.~\eqref{eq:ml} and  Eq.~\eqref{eq:avgarea}, we can show our claim by
\begin{alignat}{2}
	&|O_{\ell+1}| \geq \frac{|I_{\ell, P_{\ell+1}}|\cdot |\SK_{\ell+1}|}{|P_{\ell+1}|} &\iff\,&
	|O_{\ell+1}| \cdot |P_{\ell+1}| \geq |I_{\ell, P_{\ell+1}}| \cdot |\SK_{\ell+1}| \\
	\overset{\textnormal{\eqref{eq:ml}}}{\iff}\,& \frac{|\SK_{\ell+1}|+|P_{\ell+1}|}{|O_{\ell+1}|+|I_{\ell, P_{\ell+1}|}} \leq \frac{|\SK_{\ell+1}|}{|O_{\ell+1}|} &\iff\,& \mathcal{M}_\ell \geq \mathcal{M}_{\ell+1} \, .
\end{alignat}
	\end{proof}
	
In Fig.~\ref{fig}, we compare our committed compression scale-space to an uncommitted scale-space where the sparsification path is chosen in a uniformly random order. While the evolution behaves similarly at fine scales, the reconstruction error is consistently lower for the compression path. At coarse scales, the compression path even yields $\approx 60\%$ error reduction and is thus indeed task-adaptive.

\subsection{Skeletonisation Scale-Spaces for Branch Pruning}
\label{sec:shaperec}

Pruning~\cite{BLL07,ML12,Og94,SB98,SS16,TH02} has different goals than compression. It is often performed as a preprocessing step for applications such as shape recognition and matching. Here, the connectivity and structure of the skeleton are essential while the amount of skeleton points and thus the coding cost is irrelevant. 

Therefore, we design the sparsification path to preserve skeleton structure in two ways. First, our pruning has to be \textit{homotopy preserving}, i.e. it should not disconnect the medial axis. Furthermore, we take into account semantically important end- and branching points of the skeleton. We can quantify their importance by using \textit{skeleton complexity}~\cite{PB12} as a third quality measure.

To this end, we require two neighbourhood definitions. Let a 4-neighbourhood contain left, right, upper, and lower neighbours, while an 8-neighbourhood additionally contains the diagonal neighbours.

\begin{figure}[p]
	\small
	\tabcolsep4pt
	\begin{center}
		\begin{tabular}{ccc}
			original & segmentation & skeleton \\
			\includegraphics[width=0.3\textwidth]{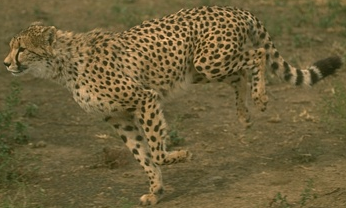} &
			\includegraphics[width=0.3\textwidth]{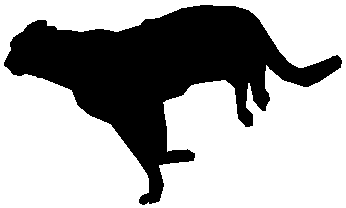} &
			\includegraphics[width=0.3\textwidth]{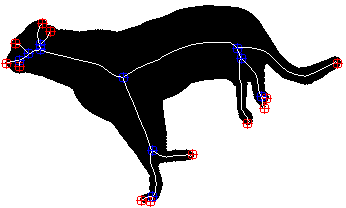} \\
			\hline
			\hline
			\textbf{random} & \textbf{compression} & \textbf{branch pruning}\\
			\includegraphics[width=0.3\textwidth]{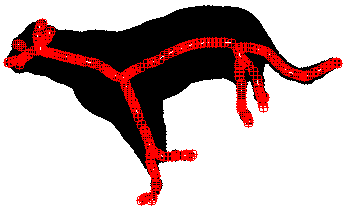} &
			\includegraphics[width=0.3\textwidth]{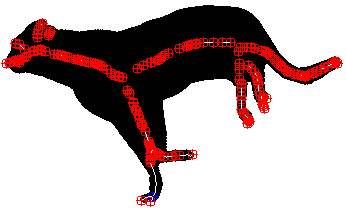} &
			\includegraphics[width=0.3\textwidth]{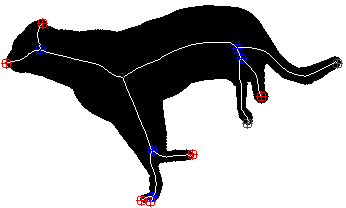} \\
			SKEL 320, ERR 1397 &  SKEL 320, ERR 1097 & SKEL 647, ERR 1174 \\
			\includegraphics[width=0.3\textwidth]{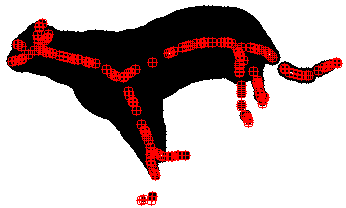} &
			\includegraphics[width=0.3\textwidth]{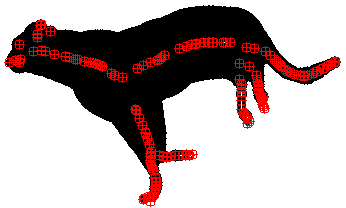} &
			\includegraphics[width=0.3\textwidth]{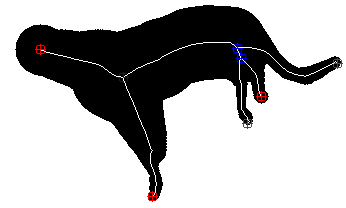} \\
			SKEL 160, ERR 2051 & SKEL 160, ERR 1225 & SKEL 530, ERR 1931 \\
			\includegraphics[width=0.3\textwidth]{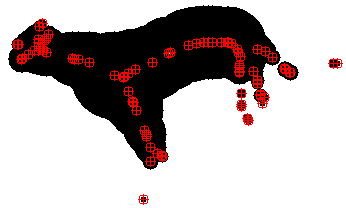} &
			\includegraphics[width=0.3\textwidth]{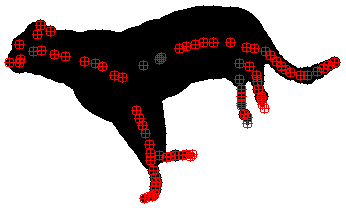} &
			\includegraphics[width=0.3\textwidth]{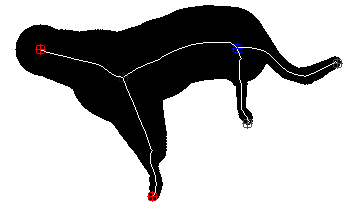} \\
			SKEL 80, ERR 3347 & SKEL 80, ERR 1440 & SKEL 492, ERR 2302 \\
			\includegraphics[width=0.3\textwidth]{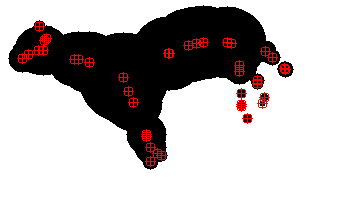} &
			\includegraphics[width=0.3\textwidth]{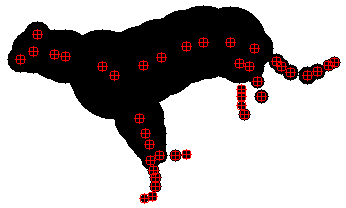} &
			\includegraphics[width=0.3\textwidth]{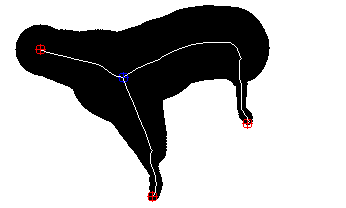} \\
			SKEL 40, ERR 4378 & SKEL 40, ERR 2147 & SKEL 392, ERR 3603 \\
			\includegraphics[width=0.3\textwidth]{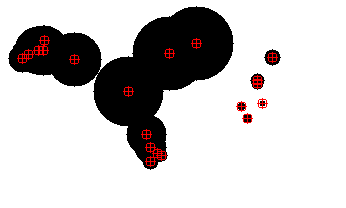} &
			\includegraphics[width=0.3\textwidth]{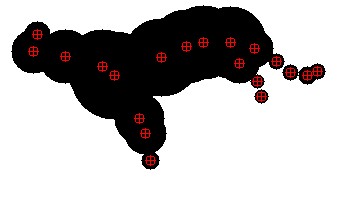} &
			\includegraphics[width=0.3\textwidth]{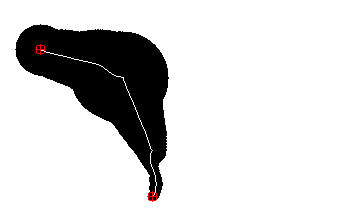} \\
			SKEL 20, ERR 8918 & SKEL 20, ERR 3803 & SKEL 203, ERR 12156\\
		\end{tabular}
	\end{center}
	\caption{\textbf{Comparison of sparsification paths} on image 134008 of BSDS500 \cite{AMFM11}. Endpoints are marked in red, branch points in blue. \textbf{SKEL} denotes the number of skeleton points and \textbf{ERR} the number of wrongly reconstructed object pixels. Random and compression scale-spaces disconnect the skeleton, while branch pruning preserves homotopy. The compression path yields consistently smaller reconstruction error than random removal.\label{fig}}
\end{figure}

\begin{definition}[Endpoints, Branching Points, Simple Points, and Arcs]\label{def:points}
	\vspace{-1.75em}
	\begin{itemize}
		\item An \textbf{endpoint} of $\SK$ has either: 1.) zero or one skeleton points in their 8-neigh\-bour\-hood; 2.) two skeleton points in its 8-neighbourhood which are adjacent to one another; 3.) three skeleton points in its 8-neighbourhood that are either all above, below, left, or right of $i$. We denote the set of endpoints by $E(\SK)$.
		\item A \textbf{branching point} of $\SK$ has at least three skeleton points in its 8-neigh\-bour\-hood which are not horizontally or vertically adjacent. We denote the set of branching points by $B(\SK)$.
		\item A \textbf{simple point} is a skeleton point which is neither an endpoint nor a branching point. We write $S(\SK) := \SK \setminus (E(\SK) \cup B(\SK))$.
		\item We define $\mathcal{A}(\SK)$ as the set of all arcs/branches. Each of its elements $A = \{a_1,...,a_k\} \subset \SK$ fulfils two properties: 1.) The arc is \textit{connected}, i.e. for each $i$ the points $a_i$ and $a_{i+1}$ share an 8-neighbourhood. 2.) Exactly $a_1$ and $a_k$ are \textit{end}- or \textit{branching points}, the rest of the $a_i$ are simple. 
		\end{itemize}
\end{definition}

Endpoints and branching points are much more important for shape recognition than simple points. They mark where individual arcs of the skeleton intersect or end. Thereby, they characterise the shape together with the arc lengths. Therefore, \textit{skeleton complexity} considers the number of important points as
\begin{equation}
	\mathcal{C}_\ell := |E(\SK_\ell)| + |B(\SK_\ell)|\, .
\end{equation}
Spurious branches resulting from noise increase this measure significantly due to added end- and branching points. Note that typically, such a change has marginal impact on the relative error, while skeleton minimality rises based on the added arc length.
With these definitions and the combined goals of preserving homotopy while not increasing skeleton complexity, we define the branch pruning path in Algorithm~\ref{alg:skelshape}. The simple rule of only removing full skeleton branches with endpoints achieves our goals. A skeleton can only be disconnected if branching or simple points are removed. We can show that complexity is decreasing for coarser scales and thus constitutes a Lyapunov sequence for the pruning scale-space.

\begin{algorithm}[t]
	\caption{Branch Pruning Path\label{alg:skelshape}}
	$\ell \gets 0$\;
	\While{$|\SK_\ell| > 0$}{
		\tcc{Consider branches with endpoints or full skeleton.}
		$A_\mathcal{E} \gets \{ A \subset \mathcal{A}(\SK_\ell)\, | \, E(\SK_\ell) \cap A \neq \emptyset \}$, \lIf{$A_\mathcal{E} = \emptyset$}{$A_\mathcal{E} \gets \{\SK_\ell\}$}
		\tcc{Select branch with smallest reconstruction impact.}
		$A_{\min} \gets \textnormal{argmin}_{A \in \mathcal{A}_E} I_{\ell, A}$ \;
		\tcc{Remove whole branch excluding branching points.}
		$P_{\ell+1}  \gets \{a_i \in A_{\min} \, | \, a_i \notin B(\SK_\ell) \}$,
		$\SK_{\ell+1} \gets \SK_\ell \setminus P_{\ell+1}$, 
		$\ell \gets \ell + 1$\;
	}	
\end{algorithm}

\begin{proposition}[Skeleton Complexity is a Lyapunov Sequence]
	Skeleton complexity decreases with increasing scale $\ell$, i.e. $\mathcal{C}_\ell \geq \mathcal{C}_{\ell+1}$.
\end{proposition}

	\begin{proof} For a scale $\ell \in \{0,...,m-1\}$, there are the following possible cases for the next step $P_{\ell+1}$ in the sparsification path:
		\begin{description}
			\item[Case 1:] $P_{\ell+1}$ has exactly two endpoints, which are both removed. The branch is not connected to other skeleton components and thus cannot influence other points. Then $\mathcal{C}_{\ell+1} = \mathcal{C}_{\ell}-2 < \mathcal{C}_{\ell}$.
			\item[Case 2:] By Algorithm~\ref{alg:skelshape}, $P_{\ell+1}$ has exactly one endpoint, which is removed. The arc is connected to a branching point. Its role can change to an end or simple point, but this cannot increase complexity. Thus, $\mathcal{C}_{\ell+1} \leq \mathcal{C}_{\ell}-1 < \mathcal{C}_{\ell}$.
			\item[Case 3:] $P_{\ell+1} = \SK_\ell$. The remaining skeleton is removed, thus $\mathcal{C}_{\ell+1} = 0 < \mathcal{C}_{\ell}$. \qed
		\end{description} 
	\end{proof}
	
In Fig.~\ref{fig},  branch pruning differs significantly from uncommitted random and compression scale-spaces. Note that visualising a meaningful evolution requires different scales than the other paths. Throughout the evolution, the branch pruning gradually reduces complexity: The joint number of end- and branching points decreases as desired. Homotopy is also preserved while the reconstruction error is kept as low as possible under these constraints. However, preserving the structure also comes at a cost. With an order of magnitude less points, the compression path yields only a quarter of the reconstruction error. On the flips-side, it sacrifices homotopy and creates many new end points.

\section{Conclusions and Outlook}
\label{sec:conclusion}
 
With skeletonisation scale-spaces, we provide a new hierarchical perspective on shape analysis. They combine the principles of the medial axis transform \cite{Bl67} with sparsification scale-space paradigms \cite{CPW19,Pe21}. Our scale-space framework allows adaptions to different applications via committed sparsification paths. In particular, our compression and branch pruning examples demonstrate that suitable requirements result in theoretical guarantees for reconstruction quality or skeleton structure in terms of Lyapunov sequences.

In the future, we plan to investigate connections between the skeleton evolution and the implied boundary evolution. Moreover, extending our theory to the fully continuous setting could yield new insights. Finally, we would like to explore full-fledged compression or shape matching approaches based on scale-spaces.

\medskip
%
%
%
\bibliographystyle{splncs04}
\bibliography{../../bibtex/bib.bib}

\end{document}